\documentclass[journal,onecolumn,12pt]{IEEEtran} 
\usepackage[]{times} 
\usepackage{epsfig} 
\usepackage{subfigure}
\usepackage{amsmath} 
\usepackage{amssymb}  
\usepackage{graphicx} 
\usepackage{subfigure}
\usepackage{multirow} 
\renewcommand{\baselinestretch}{1.5} 
\newcommand{\bega}{\begin{eqnarray}}
\newcommand{\ega}{\end{eqnarray}}
\newcommand{\bb}{\begin{equation}}
\newcommand{\ee}{\end{equation}}
\newtheorem{defn} {Definition}
\newtheorem{te}{Theorem}
\newtheorem{lema}{Lemma}
\newtheorem{cor}{Corollary}

\begin{document}
\title{On Trapping Sets and Guaranteed Error Correction Capability of LDPC Codes and GLDPC Codes}
\author{Shashi~Kiran~Chilappagari,~\IEEEmembership{Student~Member,~IEEE,}~Dung~Viet~Nguyen,~\IEEEmembership{Student~Member,~IEEE,}~Bane~Vasic,~\IEEEmembership{Senior Member,~IEEE,}~and~Michael~W.~Marcellin,~\IEEEmembership{Fellow,~IEEE}
\thanks{Manuscript received \today. This work is funded by NSF under Grant CCF-0634969, ECCS-0725405, ITR-0325979 and by the INSIC-EHDR program.}
\thanks{S. K. Chilappagari, D. V. Nguyen, B. Vasic and M. W. Marcellin are with the Department of Electrical and Computer Engineering, University of Arizona, Tucson, Arizona, 85721 USA. (emails: \{shashic, nguyendv, vasic, marcellin\}@ece.arizona.edu.}
\thanks{Parts of this work have been accepted for presentation at the International Symposium on Information Theory (ISIT'08) and the International Telemetering Conference (ITC'08).}
}
\markboth{Submitted to  IEEE Transactions on Information Theory, May 2008}%
{Submitted to IEEE Transactions on Information Theory, May 2008}
\maketitle
\vspace{-0.5in}
\begin{abstract}
The relation between the girth and the guaranteed error correction capability of $\gamma$-left regular LDPC codes when decoded using the bit flipping (serial and parallel) algorithms is investigated. A lower bound on the size of variable node sets which expand by a factor of at least $3 \gamma/4$ is found based on the Moore bound. An upper bound on the guaranteed error correction capability is established by studying the sizes of smallest possible trapping sets. The results are extended to generalized LDPC codes. It is shown that generalized LDPC codes can correct a linear fraction of errors under the parallel bit flipping algorithm when the underlying Tanner graph is a good expander. It is also shown that the bound cannot be improved when $\gamma$ is even by studying a class of trapping sets. A lower bound on the size of variable node sets which have the required expansion is established. 
\end{abstract}
\begin{center} \textbf{\small Index Terms}
\end{center}

{\small Low-density parity-check codes, bit flipping algorithms, trapping sets, error correction capability}

\section{Introduction}\label{section1}
Iterative algorithms for decoding low-density parity-check (LDPC) codes \cite{gallager} have been the focus of research over the past decade and most of their properties are well understood \cite{richardsonurbanke,richardsonurbankeshokrollahi}.  These algorithms operate by passing messages along the edges of a graphical representation of the code known as the Tanner graph, and are optimal when the underlying graph is a tree. Message passing decoders perform remarkably well which can be attributed to their ability to correct errors beyond the traditional bounded distance decoding capability. However, in contrast to bounded distance decoders (BDDs), the guaranteed error correction capability of iterative decoders is largely unknown.

The problem of recovering from a fixed number of erasures is solved for iterative decoding on the binary erasure channel (BEC). If the  size of the minimum stopping set in the Tanner graph of a code is at least $t+1$, then the decoder is guaranteed to recover from any $t$ erasures. Orlitsky \textit{et al.} \cite{orlitsky} studied the relation between stopping sets and girth and derived bounds on the smallest stopping set in any $d$-left regular Tanner graph with girth $g$.

An analogous result does not exist for decoding on other channels such as the binary symmetric channel (BSC) and the additive white Gaussian noise (AWGN) channel. In this paper, we present such a result for hard decision decoding algorithms. Gallager \cite{gallager} proposed two binary message passing algorithms, namely Gallager A and Gallager B, for decoding over the BSC. He showed that for the column-weight $\gamma \geq 3$ and $\rho >\gamma$, there exist $(n,\gamma,\rho)$ \footnote{Precise definitions will be given in Section \ref{section2} and we follow standard terminology from \cite{gallager} and \cite{tanner}} regular LDPC codes for which the bit error probability asymptotically tends to zero  whenever we operate below the threshold. The minimum distance was shown to increase linearly with the code length, but correction of a linear fraction of errors was not shown. Zyablov and Pinsker \cite{zyablov} analyzed LDPC codes under a simpler decoding algorithm known as the bit flipping algorithm, and showed that almost all the codes in the regular ensemble with $\gamma \geq 5$ can correct a constant fraction of worst case errors. Sipser and Spielman \cite{spielman} used expander graph arguments to analyze two bit flipping algorithms, serial and parallel. Specifically, they showed that these algorithms can correct a fraction of errors if the underlying Tanner graph is a good expander. Burshtein and Miller \cite{burshtein} applied expander based arguments to show that message passing algorithms can also correct a fixed fraction of worst case errors when the degree of each variable node is more than five. Feldman \textit{et al.} \cite{feldman} showed that the linear programming decoder \cite{feldman2} is also capable of correcting a fraction of errors. Recently, Burshtein \cite{burshteinisitpaper} showed that regular codes with variable nodes of degree four  are capable of correcting a linear number of errors under bit flipping algorithm. He also showed tremendous improvement in the fraction of correctable errors when the variable node degree  is at least five.

Tanner \cite{tanner} studied a class of codes constructed based on bipartite graphs and short error  correcting codes. Tanner's work is a generalization of the LDPC codes proposed by Gallager \cite{gallager} and hence these codes are referred to as generalized LDPC (GLDPC) codes. Tanner proposed code construction techniques, decoding algorithms and complexity and performance analysis to analyze these codes and derived bounds on the rate and minimum distance for these codes. Sipser and Spielman \cite{spielman} analyzed a special case of GLDPC codes (which they termed as expander codes) using expansion arguments and proposed explicit constructions of asymptotically good codes capable of correcting a fraction of errors. Zemor \cite{zemor} improved the fraction of correctable errors under a modified decoding algorithm. Barg and Zemor in \cite{barg} analyzed the error exponents of expander codes and showed that expander codes achieve capacity over the BSC. Janwa and Lal \cite{janwa} studied GLDPC codes in the most general setting by considering unbalanced bipartite graphs. Miladinovic and Fossorier \cite{fossorier} derived bounds on the guaranteed error correction capability of GLDPC codes for the special case of failures only decoding.   

The focus of this paper is to establish lower and upper bounds on the guaranteed error correction capability of LDPC codes and GLDPC codes as a function of their column-weight and girth. For the case of GLDPC codes, we also find the expansion required to guarantee correction of a fraction of errors under the parallel bit flipping algorithm, as a function of the error correction capability of the sub-code. Our approach can be summarized as follows: (a) to establish lower bounds, we determine the size of variable node sets in a left regular Tanner graph which are guaranteed to have the expansion required by bit flipping algorithms, based on the Moore bound \cite[p.180]{biggs}  and (b) to find upper bounds, we study the sizes of smallest possible trapping sets \cite{rich} in a left regular Tanner graph. 

It is well known that a random graph is a good expander with high probability \cite{spielman}. However, the fraction of nodes having the required expansion is very small and hence the code length to guarantee correction of a fixed number of errors must be large. Moreover, determining the expansion of a given graph is known to be NP hard \cite{alon}, and spectral gap methods cannot guarantee an expansion factor of more than $1/2$ \cite{spielman}. On the other hand, code parameters such as column weight and girth can be easily determined or are assumed to be known for the code under consideration. We prove that for a given column-weight, the error correction capability grows exponentially in girth. However, we note that since the girth grows logarithmically in the code length, this result does not show that the bit flipping algorithms can correct a linear fraction of errors.
    
To find an upper bound on the number of correctable errors, we study the size of sets of variable nodes which lead to decoding failures. A decoding failure is said to have occurred if the output of the decoder is not equal to the transmitted codeword \cite{rich}. The conditions that lead to decoding failures are well understood for a variety of decoding algorithms such as maximum likelihood decoding, bounded distance decoding and iterative decoding on the BEC. However, for iterative decoding on the BSC and AWGN channel, the understanding is far from complete. Two approaches have been taken in this direction, namely trapping sets \cite{rich} and pseudo-codewords \cite{koetter}. We adopt the trapping set approach in this paper to characterize decoding failures. Richardson \cite{rich} introduced the notion of trapping sets to estimate the error floor on the AWGN channel. In \cite{chilappagarione}, trapping sets were used to estimate the frame error rate of column-weigh-three LDPC codes.  In this paper, we define trapping sets with the help of fixed points for the bit flipping algorithms (both serial and parallel). We then find bounds on the size of trapping sets based on extremal graphs known as cage graphs \cite{cage}, thereby finding an upper bound on the guaranteed error correction capability. 
By saying that a code with column weight $\gamma$ and girth $2g'$ is not guaranteed to correct $k$ errors, we mean that there exists a code with column weight $\gamma$ and girth $2g'$ that fails to correct $k$ errors.

The rest of the paper is organized as follows. In Section \ref{section2}, we provide a brief introduction to LDPC codes, decoding algorithms and trapping sets \cite{rich}. In Section \ref{section3}, we prove our main theorem relating the column weight and girth to the size of variable node sets which expand by a factor of at least $3 \gamma/4$. We derive bounds on the size of trapping sets based on cage graphs in Section \ref{section4}. In Section \ref{section5}, we prove that the parallel bit flipping algorithm can correct a fraction of errors if the underlying Tanner graph is a good expander. We conclude with a few remarks in Section \ref{section6}. 

\section{Preliminaries}\label{section2}
In this section, we first establish the notation and then proceed to give a brief introduction to LDPC codes and hard decision decoding algorithms. We then give the relation between the error correction capability of the code and the expansion of the underlying Tanner graph. We finally describe trapping sets for the algorithms. 
\subsection{Graph Theory Notation}
We adopt the standard notation in graph theory (see \cite{bollobas} for example).
$G=(U,E)$ denotes a graph with set of nodes $U$ and set of edges $E$. When there is no ambiguity, we simply denote the graph by $G$. An edge $e$ is an unordered pair $(u_1,u_2)$ of nodes and is said to be incident on $u_1$ and $u_2$. Two nodes $u_1$ and $u_2$ are said to be adjacent (neighbors) if there is an edge $e=(u_1,u_2)$ incident on them. The order of the graph is $|U|$ and the size of the graph is $|E|$. The degree of $u$, $d(u)$, is the number of its neighbors. A node with degree one is called a leaf or a pendant node. A graph is $d$-regular if all the nodes have degree $d$. The average degree $\overline{d}$ of a graph is defined as $\overline{d}=2|E|/|U|$. The girth $g(G)$ of a graph $G$, is the length of smallest cycle in $G$. $H=(V \cup C,E')$ denotes a bipartite graph with two sets of nodes; variable (left) nodes $V$ and check (right) nodes $C$ and edge set $E'$. Nodes in $V$ have neighbors only in $C$ and vice versa. A bipartite graph is said to be $\gamma$-left regular if all variable nodes have degree $\gamma$, $\rho$-right regular if all check nodes have degree $\rho$ and $(\gamma,\rho)$ regular if all variable nodes have degree $\gamma$ and all check nodes have degree $\rho$. The girth of a bipartite graph is even. 
\subsection{LDPC Codes and Decoding Algorithms}
LDPC codes \cite{gallager} are a class of linear block codes which can be defined by sparse bipartite graphs \cite{shokrollahi}. Let $G$ be a bipartite graph with two sets of nodes: $n$ variable nodes and $m$ check nodes. This graph defines a linear block code $\mathcal{C}$ of length $n$ and dimension at least $n-m$ in the following way: The $n$ variable nodes are associated to the $n$ coordinates of codewords. A vector $\mathbf{v}=(v_1,v_2,\ldots,v_n)$ is a codeword if and only if for each check node, the modulo two sum  of its neighbors is zero.  Such a graphical representation of an LDPC code is called the Tanner graph \cite{tanner} of the code. The adjacency matrix of $G$ gives a parity check matrix of $\cal{C}$. An $(n,\gamma,\rho)$ regular LDPC code has a Tanner graph with $n$ variable nodes each of degree  $\gamma$ (column weight) and $n\gamma/ \rho$ check nodes each of degree  $\rho$ (row weight). This code has length $n$ and rate $r \geq 1-\gamma/\rho$ \cite{shokrollahi}.

We now describe a simple hard decision decoding algorithm known as the parallel bit flipping algorithm \cite{zyablov,spielman} to decode LDPC codes. As noted earlier, each check node imposes a constraint on the neighboring variable nodes. A constraint (check node) is said to be satisfied by a setting of variable nodes if the sum of the variable nodes in the constraint is even; otherwise the constraint is unsatisfied. 

{\bfseries Parallel Bit Flipping Algorithm}
\begin{itemize}
\item In parallel, flip each variable that is in more unsatisfied than satisfied constraints.
\item Repeat until no such variable remains. 
\end{itemize}
A serial version of the algorithm is also defined in \cite{spielman} and all the results in this paper hold for the serial bit flipping algorithm also. The bit flipping algorithms are iterative in nature but do not belong to the class of message passing algorithms (see \cite{burshtein} for an explanation).

\subsection{Expansion and Error Correction Capability}
Sipser and Spielman \cite{spielman} analyzed the performance of the  bit flipping algorithms using the expansion properties of the underlying Tanner graph of the code. We summarize the results from \cite{spielman} below for the sake of completeness. We start with the following definitions from \cite{spielman}.

\begin{defn}
Let $G=(U,E)$ with $|U|=n_1$. Then \textit{every set of at most $m_1$ nodes expands by a factor of $\delta$} if, for all sets $S \subset U$
\[
|S|\leq m_1 \Rightarrow |\{y: \exists x \in S \mbox{~such that~} (x,y) \in E \}| > \delta |S|.
\]
\end{defn}
We consider bipartite graphs and expansion of variable nodes only. 
\begin{defn}
A graph is a $(\gamma,\rho,\alpha,\delta)$ expander if it is a $(\gamma,\rho)$ regular bipartite graph in which every subset of at most $\alpha$ fraction of the variable nodes expands by a factor of at least $\delta$.
\end{defn}
The following theorem from \cite{spielman} relates the expansion and error correction capability of an $(n,\gamma,\rho)$ LDPC code with Tanner graph $G$ when decoded using the parallel bit flipping decoding algorithm.
\begin{te}\cite[Theorem 11]{spielman}
Let $G$ be a $(\gamma, \rho, \alpha, (3/4 +\epsilon)\gamma)$ expander over $n$ variable nodes, for any $\epsilon > 0$. Then, the simple parallel decoding algorithm will correct any $\alpha_0 < \alpha(1 + 4\epsilon)/2$ fraction of errors after $\log_{1-4\epsilon}(\alpha_0 n)$ decoding rounds. 
\end{te}
\textit{Notes:}
\begin{enumerate}
\item The serial bit flipping algorithm can also correct $\alpha_0 < \alpha/2$ fraction of errors if $G$ is a $(\gamma, \rho, \alpha, (3/4)\gamma)$ expander.
\item The results hold for any left regular code as expansion is needed for variable nodes only.
\end{enumerate}
From the above discussion, it is observed that finding the number of variable nodes which are guaranteed to expand by a factor of at least $3 \gamma/4$, gives a lower bound on the guaranteed error correction capability of LDPC codes.
\subsection{Decoding Failures and Trapping Sets}
We now characterize  failures of the iterative decoders using fixed points and trapping sets. Some of the following discussion appears in \cite{colwtthreepaper}, \cite{chilappagarione}, \cite{ucsdpaper} and we include it for sake of completeness. 

Consider an LDPC code of length $n$ and let $\mathbf{x}=(x_1 x_2 \ldots x_n)$ be the binary vector which is the input to the iterative decoder. Let $S(\mathbf{x})$ be the support of $\mathbf{x}$. The support of $\mathbf{x}$ is defined as the set of all positions $i$ where $x_i \neq 0$. The set of variable nodes (bits) which differ from their correct value are referred to as corrupt variables. 

\begin{defn}\cite{colwtthreepaper}
A decoder failure is said to have occurred if the output of the decoder is not equal to the transmitted codeword.
\end{defn} 

\begin{defn}
$\mathbf{x}$ is a fixed point of the bit flipping algorithm if the set of corrupt variables remains unchanged after one round of decoding.
\end{defn}

\begin{defn}\cite{chilappagarione}
The support of a fixed point is known as a trapping set. A $(V,C)$ trapping set $\cal{T}$ is a set of $V$ variable nodes whose induced subgraph has $C$ odd degree checks. 
\end{defn}

If the variable nodes corresponding to a trapping set are in error, then a decoder failure occurs. However, not all variable nodes corresponding to  a trapping set need to be in error for a decoder failure to occur.

\begin{defn}\cite{chilappagarione} The minimal number of variable  nodes that have to be initially in error for the decoder to end up in the trapping set $\cal{T}$ will be referred to as {\it critical number} $m$ for that trapping set.\end{defn}
\begin{defn} \cite{colwtthreepaper} A set of variable nodes which if in error lead to a decoding failure is known as a \textit{failure set}.\end{defn}

\section{Column Weight, Girth and Expansion}\label{section3}
In this section, we prove our main theorem which relates the column weight and girth of a code to its error correction capability. We show that the size of variable node sets which have the required expansion is related to the well known Moore bound \cite[p.180]{biggs}. We start with a few definitions required to establish the main theorem.
\subsection{Definitions}
\begin{defn}The \textit{reduced graph} $H_r=(V \cup C_r, E'_r)$ of $H=(V \cup C,E')$ is a graph with vertex set $V \cup C_r$ and edge set $E'_r$ given by 
\begin{eqnarray}
C_r &=& C \setminus C_p, ~C_p =\{c \in C : \mbox{c is a pendant node}\} \nonumber \\
E'_r&=& E' \setminus E'_p,~ E'_p = \{(v_i,c_j) \in E : c_j \in C_p\}. \nonumber 
\end{eqnarray}

\end{defn}
\begin{defn} Let $H=(V \cup C, E')$ be such that $\forall v \in V, d(v) \leq \gamma$. The \textit{$\gamma$ augmented graph} $H_{\gamma}=(V \cup C_{\gamma}, E'_{\gamma})$ is a graph with vertex set $V \cup C_{\gamma}$ and edge set $E'_{\gamma}$ given by
\begin{eqnarray}
C_{\gamma} &=& C \cup C_a, \mbox{~where~} C_a = \bigcup_{i=1}^{|V|}C_a^i \mbox{~and~} \nonumber \\
C_a^i &=& \{c_1^i,\ldots,c_{\gamma-d(v_i)}^i\}; \nonumber \\
E'_{\gamma}&=& E' \cup E'_{a}, \mbox{~where~} E'_a = \bigcup_{i=1}^{|V|} E_{a}^{'i} \mbox{~and} \nonumber \\
E_a^{'i} &=& \{(v_i,c_j)\in V \times C_{a}: c_j \in C_a^i\}. \nonumber
\end{eqnarray}
\end{defn}
  
\begin{defn}\cite[Definition 4]{spielman} The \textit{edge-vertex incidence graph} $G_{ev}=(U \cup E, E_{ev})$ of $G=(U,E)$ is the bipartite graph with vertex set $U \cup E$  and edge set
\[
E_{ev}=\{(e,u) \in E \times U : \mbox{$u$ is an endpoint of e}\}.
\]\end{defn}
\textit{Notes:}  
\begin{enumerate}
\item The edge-vertex incidence graph is right regular with degree two.
\item $|E_{ev}|=2|E|$.
\item $g(G_{ev})=2g(G)$.
\end{enumerate}

\begin{defn}An \textit{inverse edge-vertex incidence graph} $H_{iev}=(V, E'_{iev})$ of $H=(V \cup C, E')$ is a graph with vertex set $V$ and edge set $E'_{iev}$ which is obtained as follows. For $c \in C_r$, let $N(c)$ denote the set of neighbors of $c$. Label one node $v_i \in N(c)$ as a root node. Then 
\begin{eqnarray}
E'_{iev}&=&\{(v_i,v_j) \in V \times V: v_i \in N(c), v_j \in N(c),  \nonumber \\
& &i \neq j,\mbox { $v_i$ is a root node, for some $c \in C_r$} \}. \nonumber
\end{eqnarray}
\end{defn}
\textit{Notes:}  
\begin{enumerate}
\item Given a graph, the inverse edge-vertex incidence graph is not unique. 
\item $g(H_{iev}) \geq g(H)/2$, $|E'_{iev}| = |E'_r| - |C_r|$ and $|C_r| \leq |E'_r|/2$.
\item $|E'_{iev}| \geq |E'_r|/2$ with equality only if all checks in $C_r$ have degree two.
\item The term inverse edge-vertex incidence is used for the following reason. Suppose all checks in $H$ have degree two. Then the edge-vertex incidence graph of $H_{iev}$ is $H$.
\end{enumerate}

The \textit{Moore bound} \cite[p.180]{biggs} denoted by $n_0(d,g)$ is a lower bound on the least number of vertices in a $d$-regular graph with girth $g$. It is given by 
\begin{eqnarray}
n_0(d,g)=n_0(d,2r+1) &=& 1 + d \sum_{i=0}^{r-1} (d-1)^i, ~g~\mbox{odd} \nonumber\\
n_0(d,g)=n_0(d,2r)&=& 2 \sum_{i=0}^{r-1}(d-1)^i \nonumber, ~g~\mbox{even}.
\end{eqnarray}

In \cite{mooreirreg}, it was shown that a similar bound holds for irregular graphs. 
\begin{te}\cite{mooreirreg}
The number of nodes $n(\overline{d},g)$ in a graph of girth $g$ and average degree at least $\overline{d} \geq 2$ satisfies
\[
n(\overline{d},g) \geq n_0(\overline{d},g).
\]
\end{te}
Note that $\overline{d}$ need not be an integer in the above theorem.

\subsection{The Main Theorem}
We now state and prove the main theorem. 
\begin{te}\label{thm1}
Let $G$ be a $\gamma \geq 4$-left regular Tanner graph $G$ with $g(G)=2g'$. Then for all $k < n_0(\gamma/2,g')$, any set of $k$ variable nodes in $G$ expands by a factor of at least $3 \gamma/4$.\end{te}
\begin{proof}
Let $G^{k}=(V^k \cup C^k, E^k )$ denote the subgraph induced by a set of $k$ variable nodes $V^{k}$. Since $G$ is $\gamma$-left regular, $|E^{k}|=\gamma k$. Let $G^{k}_r=(V^{k} \cup C^{k}_r ,E^{k}_r)$ be the reduced graph. We have
\begin{eqnarray}
|C^{k}| &=& |C^{k}_r| + |C^{k}_p| \nonumber \\
|E^k| &=& |E^k_p| + |E^k_r| \nonumber \\
|E^k_p| &=& |C^{k}_p| \nonumber \\ 
|C^{k}_p| &=& \gamma k - |E^{k}_r|. \nonumber 
\end{eqnarray}
We need to prove that $|C^k| > 3\gamma k/4$. 

Let $f(k,g')$ denote the maximum number of edges in an arbitrary graph of order $k$ and girth $g'$. By Theorem 2, for all $k < n_0 (\gamma/2,g')$, the average degree of a graph with $k$ nodes and girth $g'$ is less than $\gamma/2$. Hence, $f(k,g') < \gamma k/4$.  We now have the following lemma.

\begin{lema}
The number of edges in  $G^{k}_r$ cannot exceed $2f(k,g')$ i.e.,
\[
|E^{k}_r| \leq  2 f(k,g').
\]
\end{lema}
\begin{proof}
The proof is by contradiction. Assume that $|E^{k}_r| > 2f(k,g')$. Consider $G^{k}_{iev}=(V^{k}, E^{k}_{iev})$, an inverse edge vertex incidence graph of $G^{k}$. We have
\[
|E^{k}_{iev}| > f(k,g'). 
\]
This is a contradiction as $G^{k}_{eiv}$ is a graph of order $k$ and girth at least $g'$.
\end{proof}
We now find a lower bound on $|C^k|$  in terms of $f(k,g')$. We have the following lemma.
\begin{lema}
$|C^{k}| \geq \gamma k - f(k,g')$. 
\end{lema}
\begin{proof}
Let $|E^{k}_{r}| = 2f(k,g') - j$ for some integer $j \geq 0$. Then $|E^{k}_{p}| = \gamma k - 2f(k,g') + j$. We claim that  $|C^{k}_{r}| \geq f(k,g') + j$. To see this, we note that 
\begin{eqnarray}
|E^{k}_{iev}| &=& |E^{k}_{r}| - |C^{k}_{r}|, \mbox{~or} \nonumber \\
|C^{k}_{r}| &=& |E^{k}_{r}| - |E^{k}_{iev}|. \nonumber 
\end{eqnarray}
But
\begin{eqnarray}
|E^{k}_{iev}| &\leq& f(k,g') \nonumber \\
\Rightarrow |C^{k}_{r}| &\geq& 2f(k,g') - j - f(k,g') \nonumber \\
\Rightarrow |C^{k}_{r}| &\geq& f(k,g') - j .\nonumber
\end{eqnarray}
Hence we have,
\begin{eqnarray}
|C^{k}| &=&  |C^{k}_{r}| + |C^{k}_{p}| \nonumber \\
\Rightarrow |C^{k}| &\geq& f(k,g') - j + \gamma k - 2f(k,g') + j \nonumber \\
\Rightarrow |C^{k}| &\geq& \gamma k - f(k,g'). \nonumber
\end{eqnarray}
\end{proof}
The theorem now follows as
\[
f(k,g') < \gamma k/4 
\]
and therefore
\[
|C^{k}| > 3\gamma k/4.
\]
\end{proof}
\begin{cor}
Let $\mathcal{C}$ be an LDPC code with column-weight $\gamma \geq 4$ and girth $2g'$. Then the bit flipping algorithm can correct any error pattern of weight less than $n_0(\gamma/2,g')/2$.
\end{cor}
\section{Cage Graphs and Trapping Sets}\label{section4}
In this section, we first give necessary and sufficient conditions for a given set of variables to be a trapping set. We then proceed to define a class of interesting graphs known as cage graphs \cite{cage} and establish a relation between cage graphs and trapping sets. We then give an upper bound on the error correction capability based on the sizes of cage graphs. The proofs in this section are along the same lines as in Section \ref{section3}. Hence, we only give a sketch of the proofs.
\begin{te}\label{thm2}
Let $\mathcal{C}$ be an LDPC code with $\gamma$-left regular Tanner graph $G$. Let $\cal{T}$ be a set consisting of $V$ variable nodes with induced subgraph $\cal{I}$. Let the checks in $\cal{I}$ be partitioned into two disjoint subsets; $\cal{O}$ consisting of checks with odd degree and $\cal{E}$ consisting of checks with even degree. Then $\cal{T}$ is a trapping set for bit flipping algorithm iff : (a) Every variable node in $\cal{I}$ has at least $\left\lceil \gamma/2 \right\rceil$  neighbors in $\cal{E}$, and (b) No $\left\lfloor \gamma/2 \right\rfloor + 1$ checks of $\cal{O}$ share a  neighbor outside $\cal{I}$.
\end{te}
\begin{proof}
We first show that the conditions stated are sufficient. Let $\mathbf{x_{\mathcal{T}}}$ be the input to the bit flipping algorithm, with support $\mathcal{T}$. The only unsatisfied constraints are in $\mathcal{O}$. By the conditions of the theorem, we observe that no variable node is involved in more unsatisfied constraints than satisfied constraints. Hence, no variable node is flipped and by definition  $\mathbf{x_{\mathcal{T}}}$ is a fixed point implying that $\mathcal{T}$ is a trapping set.

To see that the conditions are necessary, observe that for $\mathbf{x}_{\mathcal{T}}$ to be a trapping set, no variable node should be involved in more unsatisfied constraints than satisfied constraints. 
\end{proof}

\textit{Remark:} Theorem \ref{thm2} is a consequence of Fact 3 from \cite{rich}.

To determine whether a given set of variables is a trapping set, it is necessary to not only know the induced subgraph but also the neighbors of the odd degree checks. However, in order to establish general bounds on the sizes of trapping sets given only the column weight and the girth, we consider only condition (a) of Theorem \ref{thm2} which is a necessary condition. A set of variable nodes satisfying condition (a) is known as a \textit{potential trapping set}. A trapping set is a potential trapping set that satisfies condition (b). Hence, a lower bound on the size of the potential trapping set is a lower bound on the size of a trapping set. It is worth noting that a potential trapping set can always be extended to a trapping set by successively adding a variable node till condition (b) is satisfied. 
\begin{defn}\cite{cage}
A $(d,g)$-\textit{cage graph}, $G(d,g)$, is a $d$-regular graph with girth $g$ having the minimum possible number of nodes.
\end{defn}
A lower bound, $n_l(d,g)$, on the number of nodes $n_c(d,g)$ in a $(d,g)$-cage graph is given by the Moore bound. An upper bound $n_u(d,g)$ on $n_c(d,g)$ (see \cite{cage} and references therein) is given by
\begin{eqnarray}
n_u(3,g)&=& \left\{\begin{array}{cl}\frac{4}{3} + \frac{29}{12}~2^{g-2} & \mbox{for g odd} \\
																	\frac{2}{3} + \frac{29}{12}~2^{g-2} & \mbox{for g even} \end{array} \right. \nonumber \\
n_u(d,g)&=& \left\{\begin{array}{cl} 2(d-1)^{g-2} & \mbox{for g odd} \\
																	 4(d-1)^{g-3}& \mbox{for g even} \end{array} \right. . \nonumber
\end{eqnarray}

\begin{te}\label{thm3}
Let $\mathcal{C}$ be an LDPC code with $\gamma$-left regular Tanner graph $G$ and girth $2g'$. Let $\mathcal{T}(\gamma,2g')$ denote the size of smallest possible potential trapping set of $\mathcal{C}$ for the bit flipping algorithm. Then,
\[
|\mathcal{T}(\gamma,2g')| = n_c(\left\lceil \gamma/2 \right\rceil,g').
\]
\end{te}
\begin{proof}
We first prove the following lemma and then exhibit a potential trapping set of size $n_c(\left\lceil \gamma/2 \right\rceil,g')$. 
\begin{lema}
$|\mathcal{T}(\gamma,2g')| \geq n_c(\left\lceil \gamma/2 \right\rceil,g')$.
\end{lema}
\begin{proof}
Let $\mathcal{T}_1$ be a trapping set with $|\mathcal{T}_1|< n_c(\left\lceil \gamma/2 \right\rceil,g')$ and let $G_1$ denote the induced subgraph of $\mathcal{T}_1$. We can construct a $(\left\lceil \gamma/2 \right\rceil,g'')$- cage graph $(g'' \geq g)$ with $|\mathcal{T}_1|< n_c(\left\lceil \gamma/2 \right\rceil,g')$ nodes by removing edges (if necessary) from the inverse edge-vertex of $G_1$ which is a contradiction.
\end{proof}
We now exhibit a potential trapping set of size $n_c(\left\lceil \gamma/2 \right\rceil,g')$. Let $G_{ev}(\left\lceil \gamma/2 \right\rceil,g')$ be the edge-vertex incidence graph of a $G(\left\lceil \gamma/2 \right\rceil,g')$. Note that $G_{ev}(\left\lceil \gamma/2 \right\rceil,g')$ is a left regular bipartite graph with $n_c(\left\lceil \gamma/2 \right\rceil,g')$ variable nodes of degree $\left\lceil \gamma/2 \right\rceil$ and all checks have degree two. Now consider $G_{ev,\gamma}(\left\lceil \gamma/2 \right\rceil,g')$, the $\gamma$ augmented graph of $G_{ev}(\left\lceil \gamma/2 \right\rceil,g')$. It can be seen that $G_{ev,\gamma}(\left\lceil \gamma/2 \right\rceil,g')$ is a potential trapping set.
\end{proof}
\begin{te}
There exists a code $\mathcal{C}$ with $\gamma$-left regular Tanner graph of girth $2g'$ which fails to correct $n_c(\left\lceil \gamma/2 \right\rceil,g')$ errors.
\end{te}
\begin{proof}
Let $G_{ev,\gamma}(\left\lceil \gamma/2 \right\rceil,g')$ be as defined in Theorem \ref{thm3}. Now construct a code $\mathcal{C}$ with column-weight $\gamma$ and girth $2g'$ starting from $G_{ev,\gamma}(\left\lceil \gamma/2 \right\rceil,g')$ such that the set of variable nodes in $G_{ev,\gamma}(\left\lceil \gamma/2 \right\rceil,g')$ also satisfies condition (b) of Theorem \ref{thm2}. Then, by Theorem \ref{thm2} and Theorem \ref{thm3}, the set of variable nodes in $G_{ev,\gamma}(\left\lceil \gamma/2 \right\rceil,g')$ with cardinality $n_c(\left\lceil \gamma/2 \right\rceil,g')$ is a trapping set and hence $\mathcal{C}$ fails to decode an error pattern of weight $n_c(\left\lceil \gamma/2 \right\rceil,g')$.
\end{proof}

\textit{Remark:} We note that for $\gamma=3$ and $\gamma=4$, the above bound is tight. Observe that for $d=2$, the Moore bound is $n_0(d,g)=g$ and that a cycle of length $2g$ with $g$ variable nodes is always a potential trapping set. In fact, for a code with $\gamma=3$ or $4$, and Tanner graph of girth greater than eight, a cycle of the smallest length is always a trapping set (see \cite{colwtthreepaper} for the proof).
\section{Generalized LDPC Codes}\label{section5}
In this section, we first consider two bit flipping decoding algorithms for GLDPC codes. We then establish a relation between expansion and error correction capability. We also establish a lower bound on the number of variable nodes that have the required expansion. We then exhibit a trapping set and as a consequence show that the bound on the required expansion cannot be improved when $\gamma$ is even. We also establish bounds on the size of trapping sets. 

We begin with the definition of GLDPC codes by adopting the terminology from expander codes \cite{spielman}.
\begin{defn}[Definition 6, \cite{spielman}]: Let $G$ be a $(\gamma,\rho)$ regular bipartite graph between $n$ variable nodes $(v_1,v_2,\ldots,v_n)$ and $n\gamma/\rho$ check nodes $(c_1,c_2,\ldots,c_{n\gamma/\rho})$. Let $b(i,j)$ be a function designed so that, for each check node $c_i$, the variables neighboring $c_i$ are $v_{b(i,1)}, v_{b(i,2)},\ldots,v_{b(i,\rho)}$. Let $\cal{S}$ be an error correcting code of block length $\rho$. The GLDPC code $\mathcal{C}(G,\mathcal{S})$ is the code of block length $n$ whose codewords are the words $(x_1,x_2,\ldots,x_n)$ such that, for $1 \leq i \leq n\gamma/\rho$, $(x_{b(i,1)},\ldots,x_{b(i,\rho)})$ is a codeword of $\mathcal{S}$.\end{defn} 

The terms column-weight, row-weight, check nodes, variable nodes and trapping sets mean the same as in case of LDPC codes. The code $\mathcal{S}$ at each check node is sometimes referred to as the sub-code.

\subsection{Decoding algorithms}
Tanner \cite{tanner} proposed different hard decision decoding algorithms to decode GLDPC codes. We now describe an iterative algorithm known as parallel bit flipping algorithm originally described in \cite{tanner}, which is employed when the sub-code is capable of correcting $t$ errors.  

\textbf{Parallel bit flipping algorithm:} 
Each decoding round consists of the following steps.
\begin{itemize}
\item A variable node sends its current estimate to check nodes. 
\item A check node performs decoding on incoming messages and finds the nearest codeword. For all variable nodes which differ from the codeword, the check node sends a flip message. If the check node does not find a unique codeword, it does not send any flip messages. 
\item A variable node flips if it receives more than $\gamma/2$ flip messages.
\end{itemize}

The set of variable nodes which differ from their correct value are known as corrupt variables. The rest of the variable nodes are referred to as correct variables. Following the algorithms, we have the following definition adopted from \cite{spielman}:
\begin{defn}
A check node is said to be \textit{confused} if it sends flip messages to correct variable nodes, or if it does not send flip message to corrupt variable nodes, or both. Otherwise, a check node is said to be \textit{helpful}.
\end{defn}
\textit{Remarks:} 
\begin{enumerate}
\item For the parallel bit flipping decoding algorithm, a check node with sub-code of minimum distance at least $d_{min}=2t+1$ can be confused only if it is connected to more than $t$ corrupt variable nodes.
\item The parallel bit flipping algorithm is different from the algorithm presented by Sipser and Spielman in \cite{spielman} for expander codes, but is similar to the algorithm proposed by Zemor in \cite{zemor}. However, we note that the codes considered in \cite{zemor} are based on $d$-regular bipartite graphs and are a special case of doubly generalized LDPC codes, where each variable node is also associated with an error correcting code.
\item Apart from helpful checks and confused checks, Sipser and Spielman defined unhelpful checks. However, our definition of confused checks includes unhelpful checks as well. 
\item Miladinovic and Fossorier in \cite{fossorier}  considered a decoding algorithm where the decoding at every check either results in correct decoding or a failure but not miscorrection. While this assumption is reasonable when the sub-code is a long code, it is not true in general. We however, point out that the methodology we adopt can be applied to this case as well.
\item The work by Sipser and Spielman \cite{spielman}, Zemor \cite{zemor}, Barg and Zemor \cite{barg} and Janwa and Lal \cite{janwa} focused on asymptotic results and explicit construction of expander codes. The proofs and constructions are based on spectral gap and as noted earlier, such methods cannot guarantee expansion factor of more than 1/2. Our proofs require a greater expansion factor. \end{enumerate}
\subsection{Expansion and Error Correction Capability}
We now prove that the above described algorithm can correct a fraction of errors if the underlying Tanner graph is a good expander.
\begin{te}\label{thm4}
Let $\mathcal{C}(G,\mathcal{S})$ be a GLDPC code with a $\gamma$-left regular Tanner graph $G$. Assume that the sub-code $\mathcal{S}$ has minimum distance at least $d_{min}=2t+1$ and is capable of correcting $t$ errors. Let $G$ be a $(\gamma,\rho,\alpha,\beta\gamma)$ expander where 
\begin{eqnarray}
1> \beta>\frac{t+2}{2(t+1)}. \nonumber
\end{eqnarray}
Then the parallel bit flipping decoding algorithm will correct any $\alpha_0 \leq \alpha$ fraction of errors.
\end{te}
\begin{proof} 
Let $n$ be the number of variable nodes in $\mathcal{C}$. Let $V$ be the set of corrupt variables at the beginning of a decoding round. Assume that $|V|\leq\alpha n$. We will show that after the decoding round, the number of corrupt variables is strictly less than $|V|$.

Let $F$ be the set of corrupt variables that fail to flip in one decoding round, and let $C$ be the set of variables that were originally uncorrupt, but which become corrupt after one decoding round. After one decoding round, the set of corrupt variables is $F\cup C$. In the worst case scenario, a confused check sends $t$ flip messages to the uncorrupt variables and no flip message to the corrupt variables. We now have the following lemma:
\begin{lema}
Let $C_k$ be the set of confused checks, then
\begin{eqnarray}
|C_k|<\frac{(1-\beta)\gamma|V|}{t}. \label{lm1}
\end{eqnarray}
\end{lema}
\begin{proof}
The total number of edges connected to the corrupt variables is $\gamma|V|$. Each confused check must have at least $t+1$ neighbors in $V$. Let S be the set of helpful checks that have at least one neighbor in $V$. Then,
\begin{eqnarray}
\gamma|V|\geq|C_k|(t+1)+|S|. \label{lm11}
\end{eqnarray}
By expansion,
\begin{eqnarray}
|S|+|C_k|>\beta\gamma|V|. \label{lm12}
\end{eqnarray}
By (\ref{lm11}) and (\ref{lm12}), we obtain
\begin{eqnarray}
|C_k|<\frac{(1-\beta)\gamma|V|}{t}. \nonumber
\end{eqnarray}
\end{proof}
We now prove that $|F\cup C|<|V|$. The proof is by contradiction. Assume that $|F\cup C|\geq|V|$. Then there exists a subset $C'\subset C$ such that $|F\cup C'|=|V|$. We observe that a variable node in $F$ can have at most $\lfloor\gamma/2\rfloor$ neighbors that are not in $C_k$. Also, a variable node in $C'$ must have at least $\lfloor\gamma/2\rfloor + 1$ neighbors in $C_k$, and hence  can have at most $\lceil\gamma/2\rceil-1$ neighbors that are not in $C_k$. Let $N(F\cup C')$ be the set of neighbors of $F\cup C'$. Then,
\begin{eqnarray}
N(F\cup C')&\leq&|C_k|+\lfloor\frac{\gamma}{2}\rfloor|F|+\left(\lceil\frac{\gamma}{2}\rceil-1\right)|C'| \nonumber\\
&<&|C_k|+\frac{\gamma}{2}|F|+\frac{\gamma}{2}|C'|=|C_k|+\frac{\gamma}{2}|V|.\label{th1}
\end{eqnarray}
Substituting (\ref{lm1}) into (\ref{th1}), we obtain
\begin{eqnarray}
N(F\cup C')<\left(\frac{1-\beta}{t}+\frac{1}{2}\right)\gamma|V|.\nonumber
\end{eqnarray}
Now
\begin{eqnarray}
&&\beta>\frac{t+2}{2(t+1)} \nonumber\\
&=>&\frac{1-\beta}{t}<\frac{2\beta-1}{2} \nonumber\\
&=>&\frac{1-\beta}{t}+\frac{1}{2}<\beta \nonumber\\
&=>&N(F\cup C')<\beta\gamma|V|\nonumber
\end{eqnarray}
which is a contradiction.
\end{proof}   

\textit{Remark:} The above theorem proves that the parallel bit flipping algorithm can correct a fraction of errors in linear number of rounds (in code length). However, if we assume an expansion of $(\beta+\epsilon)\gamma$, it can be shown that the number of errors decreases by a constant factor with every iteration resulting in convergence in logarithmic number of rounds.

The following theorem establishes a lower bound on the number of nodes in a left regular graph which expand by a factor required by the above algorithms.
\begin{te}\label{thm6}
Let $G$ be a $\gamma$-left regular bipartite graph with $g(G)=2g'$. Then for all $k < n_0(\gamma t/(t+1),g')$, any set of $k$ variable nodes in $G$ expands by a factor of at least $\beta \gamma$, where
\begin{eqnarray}
\beta = \frac{t+2}{2(t+1)}. \nonumber
\end{eqnarray}
\end{te}

\begin{proof}
The proof is similar to the proof of Theorem \ref{thm1}. Following the notation from Theorem \ref{thm1}, we note that for all $k < n_0(\gamma t/(t+1),g')$, 
\begin{eqnarray}
f(k,g') < \frac{k\gamma t}{2(t+1)}. \nonumber
\end{eqnarray}
Since $|C^k|\geq \gamma k-f(k,g')$, we have
\begin{eqnarray}
|C^k|>\frac{t+2}{2(t+1)}\gamma k. \nonumber
\end{eqnarray}
\end{proof}
Note that the above theorem holds when $\gamma t/(t+1) \geq 2$. 
\begin{cor}
Let $\mathcal{C}(G,\mathcal{S})$ be a GLDPC code with a $\gamma$-left regular Tanner graph $G$ and $g(G)=2g'$. Assume that the sub-code $\mathcal{S}$ has minimum distance at least $d_{min}=2t+1$ and is capable of correcting $t$ errors. Then the parallel bit flipping algorithm can correct any error pattern of weight less than $n_0(\gamma t/(t+1),g')$.
\end{cor}

\subsection{Trapping Sets of GLDPC Codes}
We now exhibit a trapping set for the parallel bit flipping algorithm. By examining the expansion of the trapping set, we show that the bound given in Theorem \ref{thm4} cannot be improved when $\gamma$ is even.

\begin{te}\label{thm7}
Let $\mathcal{C}$ be a GLDPC code with $\gamma$-left regular Tanner graph $G$. Let $\cal{T}$ be a set consisting of $V$ variable nodes with induced subgraph $\cal{I}$ with the following properties: (a) The degree of each check in $\cal{I}$ is either $1$ or $t+1$; (b) Each variable node in $V$ is connected to $\left\lceil \gamma/2 \right\rceil$ checks of degree $t+1$ and $\left\lfloor \gamma/2 \right\rfloor$ checks of degree $1$; and (c) No $\left\lfloor \gamma/2 \right\rfloor + 1$ checks of degree $t+1$ share a variable node outside $\cal{I}$. Then, $\cal{T}$ is a trapping set.
\end{te}
\begin{proof}
Observe that all the checks of degree $t+1$ in $\cal{I}$ are confused. Further, each confused check does not send flip messages to variable nodes in $V$. Since any variable node in $V$ is connected to $\left\lceil \gamma/2 \right\rceil$ confused checks, it remains corrupt. Also, no variable node outside $\cal{I}$ can receive more than $\left\lfloor \gamma/2 \right\rfloor$ flip messages. Hence, no variable node which is originally correct can get corrupted. By definition, $\cal{T}$ is a trapping set. 

It can be seen that the total number of checks in $\cal{I}$ is equal to $|V|(\left\lfloor \gamma/2 \right\rfloor + \left\lceil \gamma/2 \right\rceil/(t+1))$. Hence, the set of variable nodes $V$ expands by a factor of $\gamma(t+2)/(2(t+1))$ when $\gamma$ is even. Hence, the bound given in Theorem \ref{thm4} cannot be improved in this case.  
\end{proof}

For a set of variable nodes to be a trapping set, it is necessary that every variable node in the set is connected to at least $\left\lceil \gamma/2 \right\rceil$ confused checks. This observation leads to the following bound on the size of trapping sets.
\begin{te}
Let $\mathcal{C}$ be a GLDPC code with $\gamma$-left regular Tanner graph $G$ and $g(G)=2g'$. Let $n_c(d_l,d_r,2g')$ denote the number of left vertices in a $(d_l,d_r)$ regular bipartite graph of girth $2g'$. Then the size of the smallest possible trapping set of $\cal{C}$ is $n_c(\left\lceil \gamma/2 \right\rceil, t+1 ,2g')$.
\end{te}
\begin{proof}
Follows from Theorem \ref{thm3} and Theorem \ref{thm7}
\end{proof}
\begin{cor}
Let $\mathcal{C}(G,\mathcal{S})$ be a GLDPC code with a $\gamma$-left regular Tanner graph $G$ and $g(G)=2g'$. Assume that the sub-code $\mathcal{S}$ has minimum distance at least $d_{min}=2t+1$ and is capable of correcting $t$ errors. Then the parallel bit flipping algorithm cannot be guaranteed to correct all error patterns of weight greater than or equal to $n_c(\left\lceil \gamma/2 \right\rceil, t+1 ,2g')$.
\end{cor}

\section{Concluding Remarks}\label{section6}
We derived lower bounds on the guaranteed error correction capability of LDPC and GLDPC codes by finding bounds on the number of nodes that have the required expansion. The bounds depend on two important code parameters namely: column-weight and girth. Since the relations between rate, column-weight, girth and code length are well explored in the literature (see \cite{gallager, tanner} for example), bounds on the code length needed to achieve certain error correction capability can be derived for different column weights and sub-codes (for GLDPC codes). The bounds presented in the paper serve as  guidelines in choosing code parameters in practical scenarios. 

The lower bounds derived in this paper are weak. However, extremal graphs avoiding three, four and five cycles have been studied in great detail (see \cite{extremalone,extremaltwo}) and these results can be used to derive tighter bounds when the girth is eight, ten or twelve. Also, since an expansion factor of $3 \gamma/4$ is not necessary (see \cite[Theorem 24]{spielman}) for LDPC codes, it is possible that tighter lower bounds can be derived for some cases. The results can be extended to message passing algorithms as well. There is a considerable gap between the lower bounds and upper bounds on the error correction capability. Deriving lower bounds based on the sizes of trapping sets rather than expansion may possibly lead to bridging this gap. 

Our approach can be used to derive bounds on the guaranteed erasure recovery capability for iterative decoding on the BEC by finding the number of variable nodes which expand by a factor of $\gamma/2$. In \cite{orlitsky}, the bounds on the guaranteed erasure recovery capability were derived based on the size of the smallest stopping set. Both approaches give the same bounds, which also coincide with the bounds given by Tanner \cite{tanner} for the minimum distance. Results similar to the ones reported by Miladinovic and Fossorier \cite{fossorier} based on the size of generalized stopping sets can also be derived.

\end{document}